\documentclass[12pt,a4paper]{article}
\usepackage[utf8]{inputenc}
\usepackage[T1]{fontenc}
\usepackage[english]{babel}
\usepackage{amsmath,amssymb,amsthm}
\usepackage{mathtools}
\usepackage[colorlinks=true,linkcolor=blue!70!black,citecolor=green!50!black,
            urlcolor=blue!70!black]{hyperref}
\usepackage[margin=2.5cm]{geometry}
\usepackage{enumitem}
\usepackage{tcolorbox}
\usepackage{xcolor}
\usepackage{tikz}
\usepackage{booktabs}
\usetikzlibrary{positioning}
\theoremstyle{plain}
\newtheorem{theorem}{Theorem}[section]
\newtheorem{proposition}[theorem]{Proposition}
\newtheorem{lemma}[theorem]{Lemma}
\newtheorem{corollary}[theorem]{Corollary}
\theoremstyle{definition}
\newtheorem{definition}[theorem]{Definition}
\newtheorem{example}[theorem]{Example}
\newtheorem{convention}[theorem]{Convention}
\theoremstyle{remark}
\newtheorem{remark}[theorem]{Remark}
\newcommand{\prof}{\mathrm{prof}}
\newcommand{\NN}{\mathbb{N}}
\newcommand{\cL}{\mathcal{L}}

\newcommand{\cJ}{\mathcal{J}}
\title{\textbf{Observers, Symmetries, and the Hierarchy of Language Classes:\\
A Theory of Computation Parameterized by the Observer}}

\author{
  Fabio Francesco Gabriele Buono\\
  \small Independent Researcher\\
  \small \href{https://orcid.org/0009-0004-9199-2793}{ORCID: 0009-0004-9199-2793}
}

\date{\today}

\begin{document}
\maketitle

\begin{abstract}
We introduce the \emph{observational hierarchy}, a new axis of
classification for formal languages, orthogonal to the Chomsky hierarchy.
An observer is a function $O : \Sigma^* \to S$ that determines which
information about the input is accessible to a computational system.
The order-blind automaton, which perceives the input as a multiset of
symbols rather than a sequence, constitutes the paradigmatic case.
We prove that the class of languages recognisable by any machine equipped
with such an observer coincides exactly with the permutation-closed
languages. We then define a partial order on observers that induces a
hierarchy of language classes parametrised not by the computational power
of the machine, but by the structure of the observer. We prove that this
hierarchy has the structure of a partial order with a diamond-shaped
profile sub-lattice, comprising the length branch
$O_\bot \prec O_{\mathrm{len}} \prec O_{\mathrm{prof}} \prec O_\top$
and the parity branch
$O_\bot \prec O_{\mathrm{par}} \prec O_{\mathrm{prof}} \prec O_\top$,
with $O_{\mathrm{len}}$ and $O_{\mathrm{par}}$ incomparable, and an
infinite subsequence branch
$O_\bot \prec O_1 \prec O_2 \prec \cdots \prec O_\top$, both converging
to the complete observer. We prove that the observational hierarchy is
strictly incomparable with the Chomsky hierarchy, and introduce the notion
of \emph{observational complexity} of a language. We further define
observer-parametrised complexity classes $\mathbf{P}_O$ and
$\mathbf{NP}_O$, and show that computational hardness and structural
blindness are two independent phenomena. In particular,
$\mathbf{P}_{O_{\mathrm{prof}}} = \mathbf{NP}_{O_{\mathrm{prof}}}$
holds as a structural collapse strictly inside $\mathbf{P}$.
\end{abstract}

\tableofcontents
\bigskip

\section{Introduction}

The classical theory of computational complexity classifies problems
according to the power of the machine required to solve them. Finite
automata, pushdown machines, Turing machines, and their deterministic and
nondeterministic variants~\cite{chomsky1956,sipser2012}. In this framework,
the input is always a string, an ordered sequence of symbols, and full
access to that sequence is taken for granted.

This paper inverts that perspective. We ask: \emph{which languages can a
computational system recognise when it does not have access to the complete
input, but only to a projection of it?} The power of the machine becomes
secondary; what matters is the structure of the \emph{observer}, that is,
of the function that transforms the input into available information.

The motivating case is that of an observer that cannot distinguish the
temporal order of symbols, but perceives only their counts. Formally, such
an observer maps every string $x \in \Sigma^*$ to its \emph{profile}
$\prof(x) \in \NN^k$, where $k = |\Sigma|$ and the $i$-th component
counts the occurrences of symbol $a_i$. An automaton equipped in this way
does not distinguish \texttt{01} from \texttt{10}, nor \texttt{aab} from
\texttt{aba} or \texttt{baa}.

\paragraph{Main contributions.}
\begin{enumerate}[label=(\roman*)]
  \item Definition of the order-blind automaton and characterisation of
        the languages it recognises (Section~\ref{sec:modello}).
  \item Definition of the preorder of observers, the induced partial order
        of observational power, the notion of $O$-saturation, and the
        monotonicity of induced language classes
        (Section~\ref{sec:teoria}).
  \item Construction of the observational hierarchy: proof that it has the
        structure of a partial order with a diamond-shaped profile
        sub-lattice (with $O_{\mathrm{len}}$ and $O_{\mathrm{par}}$
        incomparable, both below $O_{\mathrm{prof}}$) and an infinite
        subsequence branch, with explicit separations
        (Section~\ref{sec:gerarchia}).
  \item Formal connection with Simon's hierarchy and piecewise testable
        languages; observational reformulation of Simon's theorem
        (Section~\ref{sec:simon}).
  \item Proof that the observational hierarchy is orthogonal to the
        Chomsky hierarchy (Section~\ref{sec:ortogonalita}).
  \item Definition and first properties of observational complexity
        (Section~\ref{sec:complessita}).
  \item Physical interpretation of the framework
        (Section~\ref{sec:fisica}).
  \item Definition of observer-parametrised complexity classes
        $\mathbf{P}_O$ and $\mathbf{NP}_O$; separation of computational
        hardness from structural blindness; proof that
        $\mathbf{P}_{O_{\mathrm{prof}}} \subsetneq \mathbf{P}$
        (Section~\ref{sec:pvsnp}).
\end{enumerate}

\section{The base model: the order-blind automaton}
\label{sec:modello}

Let $\Sigma = \{a_1, \dots, a_k\}$ be a finite alphabet.

\begin{definition}[Profile]
The \emph{profile} of a string $x \in \Sigma^*$ is the function
\[
  \prof(x) = (c_1, \dots, c_k) \in \NN^k
\]
where $k = |\Sigma|$ and $c_i = |x|_{a_i}$ is the number of
occurrences of $a_i$ in $x$.
\end{definition}

\begin{definition}[Permutation-closed language]
A language $L \subseteq \Sigma^*$ is \emph{permutation-closed} if
\[
  x \in L \;\Rightarrow\;
  \forall\, y \in \Sigma^*\ \text{with}\ \prof(y) = \prof(x),
  \quad y \in L.
\]
Equivalently, $L$ is permutation-closed if and only if there exists a
function $f : \NN^k \to \{0,1\}$ such that
$L = \{x \in \Sigma^* \mid f(\prof(x)) = 1\}$.
\end{definition}

\begin{definition}[Order-blind automaton]
\label{def:ob-automaton}
An \emph{order-blind automaton} is a quintuple
$A = (Q,\, \Sigma,\, \delta,\, q_0,\, F)$
where $Q$ is a finite set of states, $q_0 \in Q$ is the initial state,
$F \subseteq Q$ is the set of accepting states, and
$\delta : Q \times \NN^k \to Q$
is the transition function. The \emph{extended transition function} is
\begin{equation}\label{eq:delta-hat}
  \hat\delta(q_0, x) \coloneqq \delta(q_0,\, \prof(x)),
\end{equation}
and the \emph{recognised language} is
\begin{equation}\label{eq:L-A}
  L(A) = \{x \in \Sigma^* \mid \hat\delta(q_0, x) \in F\}.
\end{equation}
\end{definition}

\begin{remark}
Unlike a classical DFA, whose transition function has type
$Q \times \Sigma \to Q$ and processes the input symbol by symbol,
the order-blind automaton has transition function
$\delta : Q \times \NN^k \to Q$, which reads the entire profile in a
single step (equation~\eqref{eq:delta-hat}). This is intentional: the
model is designed to recognise any permutation-closed language, including
non-regular ones such as $\{x \mid |x|_0 = |x|_1\}$.
\end{remark}

\begin{theorem}[Characterisation]
\label{thm:caratterizzazione}
A language $L \subseteq \Sigma^*$ is recognisable by an order-blind
automaton if and only if it is permutation-closed.
\end{theorem}

\begin{proof}
$(\Rightarrow)$ Let $L = L(A)$ for some order-blind automaton $A$. If
$x \in L$ and $\prof(y) = \prof(x)$, then
\[
  \hat\delta(q_0, y)
  \overset{\eqref{eq:delta-hat}}{=} \delta(q_0, \prof(y))
  = \delta(q_0, \prof(x))
  \overset{\eqref{eq:delta-hat}}{=} \hat\delta(q_0, x) \in F,
\]
so $y \in L$.

$(\Leftarrow)$ Let $L$ be permutation-closed, so there exists
$f : \NN^k \to \{0,1\}$ with
$L = \{x \in \Sigma^* \mid f(\prof(x)) = 1\}$.
Construct $A$ with $Q = \{q_0, q_{\mathrm{acc}}, q_{\mathrm{rej}}\}$,
$F = \{q_{\mathrm{acc}}\}$, and
\[
  \delta(q_0, v) = \begin{cases}
    q_{\mathrm{acc}} & \text{if } f(v) = 1, \\
    q_{\mathrm{rej}} & \text{otherwise.}
  \end{cases}
\]
Then $\hat\delta(q_0, x) = \delta(q_0, \prof(x))$ by~\eqref{eq:delta-hat},
so $A$ accepts $x$ if and only if $f(\prof(x)) = 1$, giving $L(A) = L$.
\end{proof}

\begin{example}
Over $\{0,1\}$:
\begin{itemize}
  \item $L_1 = \{x \mid |x|_1 \equiv 0 \pmod{2}\}$ is
        permutation-closed: recognisable.
  \item $L_2 = \{x \mid |x|_0 > |x|_1\}$ is permutation-closed:
        recognisable.
  \item $L_3 = \{x \mid x \text{ starts with } 1\}$ is not
        permutation-closed: $\prof(10) = \prof(01) = (1,1)$ but
        $10 \in L_3$ and $01 \notin L_3$. Not recognisable.
  \item $L_4 = 0^*1^*$ is not permutation-closed:
        $\prof(01) = \prof(10) = (1,1)$ but $01 \in L_4$ and
        $10 \notin L_4$. Not recognisable.
\end{itemize}
\end{example}

\begin{remark}
\label{rem:binding}
Theorem~\ref{thm:caratterizzazione} shows that the finiteness of $Q$ is
not the binding constraint: it is the observer that determines the class.
The power of the machine cannot compensate for the weakness of the observer.
With an observer that maps everything to the same value, no machine, not
even a Turing machine, could recognise anything beyond
$\emptyset$ and $\Sigma^*$.

The transition function $\delta : Q \times \NN^k \to Q$ is not required
to be computable. The purpose of Theorem~\ref{thm:caratterizzazione} is
to characterise \emph{all} permutation-closed languages, which includes
non-recursive and non-r.e.\ languages
(Theorem~\ref{thm:due_assi}(ii)). Restricting $\delta$ to computable
functions would yield a strictly smaller class. The model is therefore a
\emph{mathematical} characterisation rather than a computational one.
\end{remark}

\begin{corollary}
\label{cor:trivial-obs}
For any machine $M$ and any language $L \notin \{\emptyset, \Sigma^*\}$,
a system equipped with the trivial observer $O_\bot : x \mapsto \star$
does not recognise $L$, regardless of the computational power of $M$.
\end{corollary}

\begin{proof}
The trivial observer identifies all strings, the only
$O_\bot$-decidable languages are $\emptyset$ and $\Sigma^*$.
Any $L \notin \{\emptyset, \Sigma^*\}$ contains some $x$ and omits some
$y$, yet $O_\bot(x) = O_\bot(y) = \star$, so no system receiving only
$O_\bot(x)$ can decide $L$ correctly on both inputs.
\end{proof}

\section{Theory of observers}
\label{sec:teoria}

\begin{definition}[Observer]
\label{def:observer}
Let $\Sigma$ be a finite alphabet and $S$ a set (the \emph{observation
space}). An \emph{observer} is a function $O : \Sigma^* \to S$.
A \emph{computational system with observer $O$} is a pair $(M, O)$
where $M$ is any computational model that receives $O(x)$ as its entire
input in place of $x$. The language recognised by $(M, O)$ is
\[
  L(M, O) = \{x \in \Sigma^* \mid M \text{ accepts } O(x)\}.
\]
\end{definition}

Every observer $O$ induces an equivalence relation on $\Sigma^*$:
\[
  x \sim_O y \;\iff\; O(x) = O(y).
\]

\begin{definition}[$O$-saturation]
\label{def:saturation}
A language $L \subseteq \Sigma^*$ is \emph{$O$-saturated} if
\[
  x \sim_O y \;\Rightarrow\; (x \in L \iff y \in L)
\]
for all $x, y \in \Sigma^*$. Equivalently, $L$ is a union of
$\sim_O$-equivalence classes. When $S$ is infinite we fix a canonical
injective encoding $\mathrm{enc} : S \to \Sigma^*$ and interpret
$\sim_O$ as the relation induced on the encoded outputs
$\mathrm{enc}(O(x))$.
\end{definition}

\begin{example}[$O_{\mathrm{prof}}$-saturated language]
\label{ex:sat-yes}
The language $L = \{x \in \{0,1\}^* \mid |x|_0 = |x|_1\}$ depends only
on $\prof(x) = (|x|_0, |x|_1)$. If $O_{\mathrm{prof}}(x) =
O_{\mathrm{prof}}(y)$ then $|x|_0 = |y|_0$ and $|x|_1 = |y|_1$, so
$x \in L \iff y \in L$. Hence $L$ is $O_{\mathrm{prof}}$-saturated.
\end{example}

\begin{example}[Non-$O_{\mathrm{prof}}$-saturated language]
\label{ex:sat-no}
The language $L_1 = \{x \in \{0,1\}^* \mid 01 \sqsubseteq x\}$ is not
$O_{\mathrm{prof}}$-saturated. The strings $01$ and $10$ satisfy
$O_{\mathrm{prof}}(01) = O_{\mathrm{prof}}(10) = (1,1)$, yet $01 \in L_1$
(the subsequence $01$ is present) and $10 \notin L_1$ (the only
length-$2$ subsequence of $10$ is $10$ itself).
\end{example}

\begin{definition}[Induced language class]
\label{def:L-O}
\[
  \cL(O) \coloneqq \{L \subseteq \Sigma^* \mid
  L \text{ is } O\text{-saturated (Definition~\ref{def:saturation})}\}.
\]
\end{definition}

\begin{convention}[Encoding for infinite observation spaces]
\label{conv:encoding}
When $S$ is infinite (e.g.\ $S = \NN^k$ for $O_{\mathrm{prof}}$, or
$S = \mathcal{P}(\Sigma^{\leq k})$ for $O_k$), we fix a canonical
injective binary encoding $\mathrm{enc} : S \to \{0,1\}^*$ with explicit
separators. All machines receiving $O(x)$ are standard Turing machines
operating on $\mathrm{enc}(O(x))$. The cost of encoding and decoding is
at most polynomial in $|\mathrm{enc}(O(x))|$ and does not affect the
polynomial-time separations in Section~\ref{sec:pvsnp}.
\end{convention}

\begin{definition}[Observational order]
\label{def:ordine}
Given two observers $O_1 : \Sigma^* \to S_1$ and
$O_2 : \Sigma^* \to S_2$, we say that $O_1$ is \emph{less powerful}
than $O_2$, written $O_1 \preceq O_2$, if there exists a function
$f : S_2 \to S_1$ such that $O_1 = f \circ O_2$. That is, every piece
of information $O_1$ can deliver is derivable from what $O_2$ delivers.
\end{definition}

\begin{proposition}
\label{prop:ordine}
The relation $\preceq$ is a preorder on observers. It becomes a partial
order when quotienting by the equivalence
$O_1 \simeq O_2 \iff O_1 \preceq O_2 \land O_2 \preceq O_1$.
\end{proposition}

\begin{proof}
\emph{Reflexivity}: $O \preceq O$ via $f = \mathrm{id}$.
\emph{Transitivity}: if $O_1 = f \circ O_2$ and $O_2 = g \circ O_3$,
then $O_1 = (f \circ g) \circ O_3$.
\emph{Antisymmetry in the quotient}: $[O_1] \preceq [O_2]$ and
$[O_2] \preceq [O_1]$ give $O_1 \simeq O_2$, hence $[O_1] = [O_2]$.
\end{proof}

\begin{convention}
We identify observers up to $\simeq$ throughout. In particular,
$\cL(O_1) = \cL(O_2)$ whenever $O_1 \simeq O_2$: applying
Proposition~\ref{prop:monotonia} to both $O_1 \preceq O_2$ and
$O_2 \preceq O_1$ gives equality.
\end{convention}

\begin{proposition}[Monotonicity]
\label{prop:monotonia}
$O_1 \preceq O_2 \;\Rightarrow\; \cL(O_1) \subseteq \cL(O_2)$.
\end{proposition}

\begin{proof}
Let $L \in \cL(O_1)$, so $L$ is $O_1$-saturated
(Definition~\ref{def:saturation}). If $O_2(x) = O_2(y)$, then
$O_1(x) = f(O_2(x)) = f(O_2(y)) = O_1(y)$, so $x \sim_{O_1} y$.
By $O_1$-saturation of $L$, $x \in L \iff y \in L$.
Hence $L$ is $O_2$-saturated, giving $L \in \cL(O_2)$.
\end{proof}

\begin{definition}[Extremes of the order]
\begin{itemize}
  \item The \emph{trivial observer} $O_\bot : x \mapsto \star$ (constant)
        satisfies $O_\bot \preceq O$ for every $O$.
        $\cL(O_\bot) = \{\emptyset, \Sigma^*\}$.
  \item The \emph{complete observer} $O_\top : x \mapsto x$ (identity)
        satisfies $O \preceq O_\top$ for every $O$.
        $\cL(O_\top)$ is the class of all languages.
\end{itemize}
\end{definition}

\begin{remark}[Connection with information theory]
Proposition~\ref{prop:monotonia} has the form of the data processing
inequality. In the stochastic setting, in which observers map strings to
probability distributions $O : \Sigma^* \to \Delta(S)$, the order
$\preceq$ becomes \emph{channel degradation} in the sense of
Blackwell~\cite{blackwell1953}: a channel $P$ is degraded with respect to
$Q$ if $P(\cdot \mid x) = \int K(\cdot \mid y)\,\mathrm{d}Q(y \mid x)$
for some stochastic kernel $K$, the stochastic analogue of $O_1 = f \circ
O_2$. The deterministic observers of this paper are the limit in which $P$
and $Q$ are point masses; the general theory of~\cite{cover2006} provides
the informational bridge. The stochastic case is an open direction
(Section~\ref{sec:complessita}).
\end{remark}

\medskip
Having isolated $O$-saturation as the structural property induced by
observational equivalence (Definition~\ref{def:saturation}), we can now
develop the observational hierarchy and, in Section~\ref{sec:pvsnp},
the complexity theoretic framework built upon it.

\section{The observational hierarchy}
\label{sec:gerarchia}

\begin{definition}[Canonical observers]
\label{def:osservatori_canonici}
Let $\Sigma = \{0,1\}$.
\begin{enumerate}[label=(\roman*)]
  \item $O_\bot(x) = \star$ \hfill (trivial observer)
  \item $O_{\mathrm{len}}(x) = |x|$ \hfill (length only)
  \item $O_{\mathrm{par}}(x) = (|x|_0 \bmod 2,\, |x|_1 \bmod 2)$
        \hfill (parities)
  \item $O_{\mathrm{prof}}(x) = (|x|_0, |x|_1)$ \hfill (full profile)
  \item $O_k(x) = \text{set of subsequences of } x \text{ of length}
        \leq k$ \hfill ($k \geq 1$)
  \item $O_\top(x) = x$ \hfill (complete observer)
\end{enumerate}
\end{definition}

\begin{remark}
The restriction to $\Sigma = \{0,1\}$ is without loss of generality for
all structural results in this section. The order relations among
$O_\bot$, $O_{\mathrm{len}}$, $O_{\mathrm{par}}$, $O_{\mathrm{prof}}$,
$O_\top$ extend immediately to any finite alphabet
$\Sigma = \{a_1, \dots, a_k\}$ by replacing $(|x|_0, |x|_1)$ with
$(|x|_{a_1}, \dots, |x|_{a_k})$; the separating witnesses carry over
with the same arguments.
\end{remark}

\begin{proposition}[Profile sub-lattice]
\label{prop:ramo_prof}
The following strict orderings hold:
\begin{align*}
  O_\bot \prec O_{\mathrm{len}} \prec O_{\mathrm{prof}} \prec O_\top
    &\tag{length branch} \\
  O_\bot \prec O_{\mathrm{par}} \prec O_{\mathrm{prof}} \prec O_\top
    &\tag{parity branch}
\end{align*}
and $O_{\mathrm{len}}$ and $O_{\mathrm{par}}$ are incomparable:
$O_{\mathrm{len}} \not\preceq O_{\mathrm{par}}$ and
$O_{\mathrm{par}} \not\preceq O_{\mathrm{len}}$.
\end{proposition}

\begin{proof}
\textbf{Length branch.}

\smallskip\noindent
\emph{$O_\bot \prec O_{\mathrm{len}}$.}
Inclusion: via $f : \NN \to \{\star\}$ constant.
Strictness: $L = \{x \mid |x| \text{ even}\}$ is in
$\cL(O_{\mathrm{len}}) \setminus \cL(O_\bot)$.

\smallskip\noindent
\emph{$O_{\mathrm{len}} \prec O_{\mathrm{prof}}$.}
Inclusion: via $f(c_0, c_1) = c_0 + c_1$, recovering $|x|$.
Strictness: $L = \{x \mid |x|_1 \equiv 0 \pmod{2}\}$ satisfies
$O_{\mathrm{len}}(01) = O_{\mathrm{len}}(11) = 2$ while $01 \notin L$
($|01|_1 = 1$, odd) and $11 \in L$ ($|11|_1 = 2$, even), so
$L \notin \cL(O_{\mathrm{len}})$.

\smallskip\noindent
\emph{$O_{\mathrm{prof}} \prec O_\top$.}
Inclusion: via $f(x) = \prof(x)$.
Strictness: $L = 0^*1^*$ satisfies $\prof(01) = \prof(10) = (1,1)$
while $01 \in L$ and $10 \notin L$
(Definition~\ref{def:saturation}), so $L \notin \cL(O_{\mathrm{prof}})$.

\medskip\noindent
\textbf{Parity branch.}

\smallskip\noindent
\emph{$O_\bot \prec O_{\mathrm{par}}$.}
Inclusion: via the constant map. Strictness: $\{x \mid |x|_1 \equiv 0
\pmod{2}\} \in \cL(O_{\mathrm{par}}) \setminus \cL(O_\bot)$.

\smallskip\noindent
\emph{$O_{\mathrm{par}} \prec O_{\mathrm{prof}}$.}
Inclusion: via $f(c_0, c_1) = (c_0 \bmod 2, c_1 \bmod 2)$.
Strictness: let $L = \{x \mid |x|_1 \geq 3\}$. The strings $011$
and $01111$ satisfy $O_{\mathrm{par}}(011) = (1, 0)$ (since
$|011|_0 = 1$, $|011|_1 = 2$) and $O_{\mathrm{par}}(01111) = (1, 0)$
(since $|01111|_0 = 1$, $|01111|_1 = 4$), yet $011 \notin L$
($|011|_1 = 2 < 3$) and $01111 \in L$ ($|01111|_1 = 4 \geq 3$).
Hence $L$ is not $O_{\mathrm{par}}$-saturated
(Definition~\ref{def:saturation}), so $L \notin \cL(O_{\mathrm{par}})$.

\smallskip\noindent
\emph{$O_{\mathrm{prof}} \prec O_\top$.}
As in the length branch.

\medskip\noindent
\textbf{Incomparability $O_{\mathrm{len}} \not\preceq O_{\mathrm{par}}$
and $O_{\mathrm{par}} \not\preceq O_{\mathrm{len}}$.}

\smallskip\noindent
\emph{$O_{\mathrm{par}} \not\preceq O_{\mathrm{len}}$.}
$O_{\mathrm{len}}(01) = O_{\mathrm{len}}(11) = 2$, but
$O_{\mathrm{par}}(01) = (1,1) \neq (0,0) = O_{\mathrm{par}}(11)$.
No $f$ with $O_{\mathrm{par}} = f \circ O_{\mathrm{len}}$ can exist.

\smallskip\noindent
\emph{$O_{\mathrm{len}} \not\preceq O_{\mathrm{par}}$.}
For any $j \geq 0$, $O_{\mathrm{par}}(0^{2j}) =
O_{\mathrm{par}}(0^{2j+2}) = (0,0)$, but
$O_{\mathrm{len}}(0^{2j}) = 2j \neq 2j+2 =
O_{\mathrm{len}}(0^{2j+2})$. No such $f$ exists.
\end{proof}

\begin{proposition}[Subsequence branch]
\label{prop:ramo_ok}
$O_\bot \prec O_1 \prec O_2 \prec \cdots \prec O_\top$.
\end{proposition}

\begin{proof}
\emph{$O_\bot \prec O_1$.} $O_\bot(0) = O_\bot(\varepsilon) = \star$
while $O_1(0) \neq O_1(\varepsilon)$, so the inclusion $\cL(O_\bot)
\subseteq \cL(O_1)$ is strict.

\medskip\noindent
\emph{$O_k \prec O_{k+1}$ for every $k \geq 1$.}
The language $L_k = \{x \mid 0^{k+1} \sqsubseteq x\}$ is in
$\cL(O_{k+1})$, since membership is determined directly by $O_{k+1}(x)$.
It is not in $\cL(O_k)$: every subsequence of $0^k$ of length $\leq k$
is some $0^j$ with $j \leq k$, which is also a subsequence of $0^{k+1}$,
and vice versa, so $O_k(0^k) = O_k(0^{k+1})$. Yet $0^k \notin L_k$
and $0^{k+1} \in L_k$, so $L_k$ is not $O_k$-saturated
(Definition~\ref{def:saturation}).

\medskip\noindent
\emph{$O_k \prec O_\top$ for every finite $k$.}
Inclusion via $f(x) = O_k(x)$. Strictness: $O_k$ identifies $0^n 1^n$
with $0^m 1^m$ for all $n, m > k$, since both strings contain the same
set of subsequences of length $\leq k$. Hence
$\{0^n 1^n \mid n \geq 0\} \notin \cL(O_k)$ for any finite $k$,
while $O_\top$ recognises it.
\end{proof}

\begin{proposition}[Cross-branch incomparability]
\label{prop:incomparabilita}
For every $k \geq 1$, $O_{\mathrm{prof}}$ and $O_k$ are incomparable.
Consequently, $O_{\mathrm{len}}$ and $O_{\mathrm{par}}$ are also
incomparable with every $O_k$.
\end{proposition}

\begin{proof}
\emph{$O_k \not\preceq O_{\mathrm{prof}}$.}
The strings $x = 0^k 1$ and $y = 10^k$ satisfy
$\prof(x) = \prof(y) = (k,1)$, so any $f$ with
$O_k = f \circ O_{\mathrm{prof}}$ would require $O_k(x) = O_k(y)$.
But $0^k 1$ is a subsequence of $x$ of length $k+1$ and not of $y$
(since in $y$ the symbol $1$ precedes all zeros), so
$O_k(x) \neq O_k(y)$: no such $f$ exists.

\smallskip\noindent
\emph{$O_{\mathrm{prof}} \not\preceq O_k$.}
$O_k(0^{k+1}) = O_k(0^{k+2})$ (same subsequences of zeros of length
$\leq k$), but $O_{\mathrm{prof}}(0^{k+1}) = (k+1,0) \neq (k+2,0) =
O_{\mathrm{prof}}(0^{k+2})$: no $f$ with
$O_{\mathrm{prof}} = f \circ O_k$ exists.

\smallskip\noindent
\emph{$O_{\mathrm{len}}$ incomparable with every $O_k$.}
$O_k \not\preceq O_{\mathrm{len}}$: $O_{\mathrm{len}}(0) =
O_{\mathrm{len}}(1) = 1$ but $O_1(0) \neq O_1(1)$.
$O_{\mathrm{len}} \not\preceq O_k$: $O_k(0^{k+1}) = O_k(0^{k+2})$
but $O_{\mathrm{len}}(0^{k+1}) \neq O_{\mathrm{len}}(0^{k+2})$.

\smallskip\noindent
\emph{$O_{\mathrm{par}}$ incomparable with every $O_k$.}
$O_{\mathrm{par}} \not\preceq O_k$: $O_k(0) = O_k(00) =
\{\varepsilon, 0\}$ for $k \geq 1$, but $O_{\mathrm{par}}(0) = (1,0)
\neq (0,0) = O_{\mathrm{par}}(00)$.
$O_k \not\preceq O_{\mathrm{par}}$: $O_k(0^{k+1}) = O_k(0^{k+2})$
but $O_{\mathrm{par}}(0^{k+1}) = ((k+1) \bmod 2, 0) \neq
((k+2) \bmod 2, 0) = O_{\mathrm{par}}(0^{k+2})$.

\smallskip\noindent
\emph{$O_{\mathrm{len}}$ and $O_{\mathrm{par}}$ are incomparable.}
This is Proposition~\ref{prop:ramo_prof}.
\end{proof}

\begin{proposition}[Complete order relations]
\label{prop:relazioni_complete}
Among the canonical observers of
Definition~\ref{def:osservatori_canonici}, the complete set of order
relations is:
\begin{enumerate}[label=(\roman*)]
  \item Length branch:
        $O_\bot \prec O_{\mathrm{len}} \prec O_{\mathrm{prof}} \prec
        O_\top$
  \item Parity branch:
        $O_\bot \prec O_{\mathrm{par}} \prec O_{\mathrm{prof}} \prec
        O_\top$
  \item $O_{\mathrm{len}}$ and $O_{\mathrm{par}}$ are incomparable
  \item Subsequence branch:
        $O_\bot \prec O_1 \prec O_2 \prec \cdots \prec O_\top$
  \item $O_{\mathrm{prof}}$ is incomparable with every $O_k$
  \item $O_{\mathrm{len}}$ and $O_{\mathrm{par}}$ are each incomparable
        with every $O_k$
  \item No further comparabilities exist among these observers
\end{enumerate}
\end{proposition}

\begin{proof}
Items (i)--(iii): Proposition~\ref{prop:ramo_prof}.
Item (iv): Proposition~\ref{prop:ramo_ok}.
Items (v)--(vi): Proposition~\ref{prop:incomparabilita}.
Item (vii): every pair has been checked in (i)--(vi) and in the
definitions of $O_\bot$ and $O_\top$.
\end{proof}

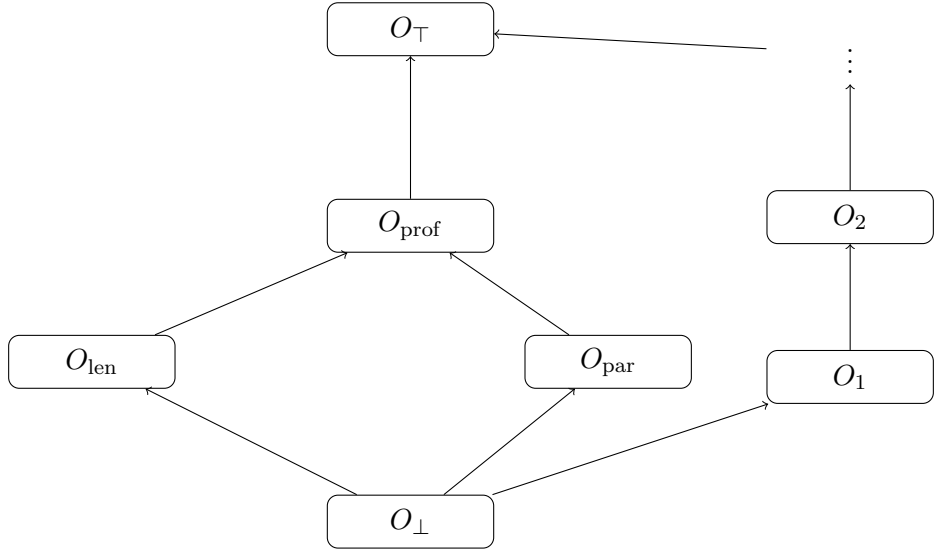
\bigskip
\begin{center}
\begin{tikzpicture}[
  node distance=1.4cm,
  every node/.style={draw, rounded corners, minimum width=2.2cm,
                     minimum height=0.7cm, align=center, font=\small}
]
  \node (bot)  {$O_\bot$};
  \node (len)  [above left=1.4cm and 2.0cm of bot]  {$O_{\mathrm{len}}$};
  \node (par)  [above right=1.4cm and 0.4cm of bot] {$O_{\mathrm{par}}$};
  \node (prof) [above=3.2cm of bot]  {$O_{\mathrm{prof}}$};
  \node (O1)   [above right=1.2cm and 3.6cm of bot] {$O_1$};
  \node (O2)   [above=of O1]   {$O_2$};
  \node (Odots)[above=of O2, draw=none] {$\vdots$};
  \node (top)  [above=5.8cm of bot] {$O_\top$};

  \draw[->] (bot)  -- (len);
  \draw[->] (bot)  -- (par);
  \draw[->] (len)  -- (prof);
  \draw[->] (par)  -- (prof);
  \draw[->] (prof) -- (top);
  \draw[->] (bot)  -- (O1);
  \draw[->] (O1)   -- (O2);
  \draw[->] (O2)   -- (Odots);
  \draw[->] (Odots) -- (top);
\end{tikzpicture}
\end{center}

\smallskip
\begin{center}
\small
The profile sub-lattice (left) forms a diamond: $O_{\mathrm{len}}$ and
$O_{\mathrm{par}}$ are incomparable, both below $O_{\mathrm{prof}}$.
The subsequence branch (right) is infinite. Neither branch is
comparable with the other at any level.
\end{center}
\bigskip

\begin{remark}
The length branch captures the total number of symbols, the parity branch
captures only the parities of individual symbol counts. The full profile
branch captures all absolute symbol counts without order. An observer
preserving partial ordering information resides in the subsequence branch.
\end{remark}

\begin{remark}[Dimensional asymmetry]
The subsequence branch is strictly infinite:
Theorem~\ref{thm:simon_sep} shows that $\cL(O_k) \subsetneq
\cL(O_{k+1})$ for every $k \geq 0$. The profile sub-lattice, by contrast, has finitely many elements between $O_\bot$ and $O_\top$ among
the canonical observers of Definition~\ref{def:osservatori_canonici}.
This asymmetry reflects the different natures of the two types of information loss: counting admits finitely many canonical coarsenings
among the observers considered here, while ordering admits a strict
infinite hierarchy. Whether further canonical observers exist strictly
between $O_\bot$ and $O_{\mathrm{prof}}$, incomparable with both
$O_{\mathrm{len}}$ and $O_{\mathrm{par}}$, is an open question
(Section~\ref{sec:complessita}).
\end{remark}

\section{Connection with Simon's hierarchy}
\label{sec:simon}

\subsection{Piecewise testable languages}

\begin{definition}[Subsequence]
A string $u$ is a \emph{subsequence} of $x$ (written $u \sqsubseteq x$)
if $u$ can be obtained by deleting symbols from $x$.
\end{definition}

\begin{definition}[Simon's preorder]
\[
  x \leq_k y \;\iff\; \forall\, u \in \Sigma^{\leq k},\;
  u \sqsubseteq x \Rightarrow u \sqsubseteq y.
\]
The corresponding equivalence is
$x \sim_k y \iff x \leq_k y \land y \leq_k x$.
\end{definition}

\begin{remark}
The relation $\sim_k$ coincides exactly with $\sim_{O_k}$: the family
$\{O_k\}$ is the computational realisation of Simon's preorder.
\end{remark}

\begin{definition}[$k$-piecewise testable language]
$L$ is \emph{$k$-piecewise testable} if $L \in \cL(O_k)$.
It is \emph{piecewise testable} if $L \in \bigcup_{k \geq 0} \cL(O_k)$.
\end{definition}

\subsection{Simon's theorem}

\begin{theorem}[Simon, 1975~\cite{simon1975}]
\label{thm:simon}
A regular language $L$ is piecewise testable if and only if its
syntactic monoid is $\cJ$-trivial.
\end{theorem}

\begin{corollary}[Observational reformulation]
\label{cor:simon_obs}
A regular language $L$ is in $\bigcup_{k} \cL(O_k)$ if and only if its
syntactic monoid is $\cJ$-trivial.
\end{corollary}

\begin{remark}
Corollary~\ref{cor:simon_obs} makes explicit that membership in the piecewise testable class is equivalent to $O_k$-decidability for some
finite $k$, a fact that the algebraic characterisation leaves implicit.
For the algebraic theory of syntactic monoids and varieties, see
also~\cite{pin1986}.
\end{remark}

\begin{theorem}[Separation of Simon's hierarchy~\cite{simon1975}]
\label{thm:simon_sep}
$\cL(O_k) \subsetneq \cL(O_{k+1})$ for every $k \geq 0$.
\end{theorem}

\begin{proof}
The language $L_k = \{x \mid 0^{k+1} \sqsubseteq x\}$ is in
$\cL(O_{k+1})$ but not in $\cL(O_k)$, since $O_k(0^k) = O_k(0^{k+1})$
(same subsequences of length $\leq k$) while $0^k \notin L_k$ and
$0^{k+1} \in L_k$: $L_k$ is not $O_k$-saturated
(Definition~\ref{def:saturation}).
\end{proof}

\subsection{Position of $O_{\mathrm{prof}}$ with respect to Simon's
hierarchy}

\begin{proposition}
\label{prop:prof_vs_simon}
$\cL(O_{\mathrm{prof}})$ and $\bigcup_k \cL(O_k)$ are incomparable.
\end{proposition}

\begin{proof}
\emph{$\cL(O_{\mathrm{prof}}) \not\subseteq \bigcup_k \cL(O_k)$.}
Let $L = \{x \in \{0,1\}^* \mid |x|_0 = |x|_1\}$.
Since $L$ depends only on the profile, $L \in \cL(O_{\mathrm{prof}})$
(Example~\ref{ex:sat-yes}). However, $L$ is not regular: by the pumping
lemma~\cite{sipser2012}, the string $0^p 1^p \in L$ cannot be pumped
within $L$ for any $p \geq 1$. Since every piecewise testable language
is regular by Theorem~\ref{thm:simon}, $L \notin \bigcup_k \cL(O_k)$.

\emph{$\bigcup_k \cL(O_k) \not\subseteq \cL(O_{\mathrm{prof}})$.}
Let $L_1 = \{x \mid 01 \sqsubseteq x\}$. Membership depends only on
whether $01$ is present as a subsequence, so $L_1 \in \cL(O_2) \subseteq
\bigcup_k \cL(O_k)$. But $L_1$ is not $O_{\mathrm{prof}}$-saturated
(Example~\ref{ex:sat-no}), so $L_1 \notin \cL(O_{\mathrm{prof}})$.
\end{proof}

Proposition~\ref{prop:prof_vs_simon} completes the picture: the profile sub-lattice and the subsequence branch represent two orthogonal dimensions
of information loss, both accessible from $O_\top$ and both invisible
from $O_\bot$.

\section{Orthogonality with the Chomsky hierarchy}
\label{sec:ortogonalita}

Let $\mathbf{REG}$, $\mathbf{CFL}$, $\mathbf{CSL}$, $\mathbf{RE}$
denote the standard Chomsky hierarchy~\cite{chomsky1956,sipser2012}, and
let $\cL_\pi$ be the class of permutation-closed languages.

\begin{remark}[Parikh's theorem]
Parikh's theorem~\cite{parikh1966} states that for every $L \in \mathbf{CFL}$ there exists $L' \in \mathbf{REG}$ with the same
Parikh image: $\{\prof(x) \mid x \in L\} = \{\prof(x) \mid x \in L'\}$.
In observational terms, $L$ and $L'$ are
$O_{\mathrm{prof}}$-profile-equivalent: no observer $O \preceq
O_{\mathrm{prof}}$ can determine from $\prof(x)$ alone whether a given
profile vector belongs to the image of $L$ but not of $L'$. This shows that $\cL(O_{\mathrm{prof}})$ cuts across the Chomsky
hierarchy in a non-trivial way.
\end{remark}

\begin{theorem}[Incomparability]
\label{thm:ortogonalita}
$\cL_\pi$ is strictly incomparable with $\mathbf{REG}$.
\end{theorem}

\begin{proof}
\emph{$\cL_\pi \not\subseteq \mathbf{REG}$}: the language
$L = \{x \in \{0,1\}^* \mid |x|_0 = |x|_1\}$ is permutation-closed
(Example~\ref{ex:sat-yes}) but not regular (pumping
lemma~\cite{sipser2012}).

\emph{$\mathbf{REG} \not\subseteq \cL_\pi$}: $0^*1^*$ is regular but not permutation-closed, since $\prof(01) = \prof(10) = (1,1)$ while
$01 \in 0^*1^*$ and $10 \notin 0^*1^*$.
\end{proof}

\begin{corollary}
The observational hierarchy and the Chomsky hierarchy are independent
axes of language classification:
\begin{center}
\begin{tabular}{ll}
  \toprule
  regular and $O_{\mathrm{prof}}$-saturated: &
    $\{x \mid |x|_1 \text{ even}\}$ \\
  regular and not $O_{\mathrm{prof}}$-saturated: &
    $0^*1^*$ \\
  non-regular and $O_{\mathrm{prof}}$-saturated: &
    $\{x \mid |x|_0 = |x|_1\}$ \\
  non-regular and not $O_{\mathrm{prof}}$-saturated: &
    $\{0^n 1^n \mid n \geq 0\}$ \\
  \bottomrule
\end{tabular}
\end{center}
\end{corollary}

\begin{proposition}[Universal incomparability]
\label{prop:universal-incomp}
For every observer $O$ with $O_{\mathrm{prof}} \preceq O \preceq O_\top$
and $O \not\simeq O_\top$, the class $\cL(O)$ is strictly incomparable
with $\mathbf{REG}$.
\end{proposition}

\begin{proof}
Since $O \not\simeq O_\top$, there exist $x \neq y$ with $O(x) = O(y)$.
The language $L_0 = \{x\}$ is in $\mathbf{REG}$ but not in $\cL(O)$:
$x \sim_O y$ while $x \in L_0$ and $y \notin L_0$, so $L_0$ is not
$O$-saturated (Definition~\ref{def:saturation}). Hence
$\mathbf{REG} \not\subseteq \cL(O)$.

For the other direction, since $O_{\mathrm{prof}} \preceq O$,
Proposition~\ref{prop:monotonia} gives $\cL(O_{\mathrm{prof}})
\subseteq \cL(O)$. By Proposition~\ref{prop:prof_vs_simon} and
Theorem~\ref{thm:ortogonalita}, the language $\{x \mid |x|_0 = |x|_1\}$
is in $\cL(O_{\mathrm{prof}}) \subseteq \cL(O)$ but not regular.
Hence $\cL(O) \not\subseteq \mathbf{REG}$.
\end{proof}

\begin{remark}
Proposition~\ref{prop:universal-incomp} covers observers in the range $O_{\mathrm{prof}} \preceq O \not\simeq O_\top$. For observers
incomparable with $O_{\mathrm{prof}}$ (such as the $O_k$), the same
conclusion holds by symmetric arguments using the witnesses in
Theorem~\ref{thm:ortogonalita} and Proposition~\ref{prop:prof_vs_simon}.
The observational hierarchy is therefore systematically transversal to
the Chomsky hierarchy at all levels.
\end{remark}

\section{Observational complexity}
\label{sec:complessita}

\begin{definition}[Observational complexity]
Let $\mathcal{F}$ be a fixed family of observers. The
\emph{observational complexity} of a language $L$ with respect to
$\mathcal{F}$ is
\[
  \mathrm{OC}_{\mathcal{F}}(L) \coloneqq
  \min_{\preceq} \{O \in \mathcal{F} \mid L \in \cL(O)\},
\]
when the minimum (in the partial order $\preceq$) exists. When no
minimum exists, the observational complexity is the antichain of minimal
observers in $\{O \in \mathcal{F} \mid L \in \cL(O)\}$.
\end{definition}

\begin{proposition}[Examples]
In the family $\mathcal{F} = \{O_\bot, O_{\mathrm{len}},
O_{\mathrm{par}}, O_{\mathrm{prof}}, O_1, O_2, \dots, O_\top\}$:
\begin{enumerate}[label=(\roman*)]
  \item $\mathrm{OC}(\Sigma^*) = \mathrm{OC}(\emptyset) = O_\bot$.
  \item $\mathrm{OC}(\{x \mid |x| \text{ even}\}) = O_{\mathrm{len}}$.
  \item $\mathrm{OC}(\{x \mid |x|_1 \equiv 0 \pmod 2\})
        = O_{\mathrm{par}}$.
  \item $\mathrm{OC}(\{x \mid |x|_0 = |x|_1\}) = O_{\mathrm{prof}}$.
  \item $\mathrm{OC}(0^*1^*) = O_\top$.
\end{enumerate}
\end{proposition}

\begin{theorem}[Computational separation]
\label{thm:separazione}
Let $O_1 \prec O_2$ strictly. Then there exists a language $L$ such that $L$ is not decidable by any machine equipped with $O_1$, regardless
of its computational power, but is decidable by a finite-state machine
equipped with $O_2$.
\end{theorem}

\begin{proof}
Since $O_1 \prec O_2$, there exist $x, y \in \Sigma^*$ with
$O_2(x) \neq O_2(y)$ but $O_1(x) = O_1(y)$.
Let $L = \{z \in \Sigma^* \mid O_2(z) = O_2(x)\}$.
Then $x \in L$ and $y \notin L$, but since $O_1(x) = O_1(y)$,
$L$ is not $O_1$-saturated (Definition~\ref{def:saturation}): no system
receiving only $O_1(z)$ can distinguish $x$ from $y$, so $L$ is not
$O_1$-decidable regardless of machine power.

Under observer $O_2$, the three-state machine
$M = (\{q_0, q_{\mathrm{acc}}, q_{\mathrm{rej}}\},\, S_2,\, \delta,\,
q_0,\, \{q_{\mathrm{acc}}\})$ with
\[
  \delta(q_0, s) = \begin{cases}
    q_{\mathrm{acc}} & \text{if } s = O_2(x), \\
    q_{\mathrm{rej}} & \text{otherwise,}
  \end{cases}
\]
accepts exactly the strings $z$ with $O_2(z) = O_2(x)$, so
$L(M, O_2) = L$.
\end{proof}

\begin{remark}
The machine $M$ operates on the codomain $S_2$ of $O_2$, which may be infinite. By Convention~\ref{conv:encoding}, this is interpreted as a
standard Turing machine operating on $\mathrm{enc}(O_2(z))$: it
compares $\mathrm{enc}(O_2(z))$ against the fixed string
$\mathrm{enc}(O_2(x))$ in linear time.
\end{remark}

\begin{corollary}[Structural limitation]
\label{cor:limite_strutturale}
There exist languages that a Turing machine with observer $O_{\mathrm{prof}}$ cannot recognise, but that a DFA with observer
$O_\top$ recognises. The power of the machine cannot compensate for the
weakness of the observer.
\end{corollary}

\section{Physical interpretation}
\label{sec:fisica}

The framework admits a natural physical reading. In statistical mechanics, a macroscopic observer has access only to aggregate quantities
such as temperature, pressure, and mean density. Formally, such an
observer is a function $O_{\mathrm{macro}} : H \to \Gamma$ where $H$ is
the space of microscopic histories and $\Gamma$ is the space of macrostates. Identifying $H = \Sigma^*$ and $\Gamma = \NN^k$, the
macroscopic observer coincides exactly with $O_{\mathrm{prof}}$.

\begin{corollary}[Observational limitation]
\label{cor:irreversibilita}
Let $O$ be an observer with $O \preceq O_{\mathrm{prof}}$. Then $O$
cannot distinguish any microscopic history $h$ from any of its
permutations. Every physical property decidable by $O$ is invariant
under permutation of events.
\end{corollary}

\begin{proof}
If $O \preceq O_{\mathrm{prof}}$, then $O = f \circ O_{\mathrm{prof}}$
for some $f$, so $O(x) = f(\prof(x)) = f(\prof(y)) = O(y)$ whenever
$\prof(x) = \prof(y)$, i.e.\ whenever $x$ and $y$ are permutations of
each other. Hence every $O$-saturated language
(Definition~\ref{def:saturation}) is permutation-closed.
\end{proof}

\begin{remark}
Corollary~\ref{cor:irreversibilita} is a restatement of
permutation-closure in physical language. The physical interpretation,
that thermodynamic irreversibility is an observational property of a
macroscopic observer structurally blind to microscopic order, is a
reading of the corollary, not a consequence derived from physical axioms.
A rigorous derivation of thermodynamic irreversibility from this framework
would require additional structure: a probability measure on microscopic
histories, a dynamics preserving the macroscopic observer, and an entropy
functional compatible with $O_{\mathrm{prof}}$.
These extensions are open directions (i) and (iii) in
Section~\ref{sec:complessita}.
\end{remark}

\begin{center}
\begin{tabular}{p{5.5cm}|p{7cm}}
  \textbf{Computer science} & \textbf{Physics} \\
  \hline
  Observer $O : \Sigma^* \to S$ & Measuring instrument \\
  Observation space $S$ & Space of macrostates \\
  $O_{\mathrm{prof}}$ & Thermodynamic quantities \\
  $O_\top$ & Complete microscopic description \\
  $O_1 \preceq O_2$ & $O_1$ coarser than $O_2$ \\
  Non-$O$-saturated language & Non-measurable property \\
  Structural limitation (Cor.~\ref{cor:limite_strutturale}) &
    Thermodynamic irreversibility \\
\end{tabular}
\end{center}

\begin{remark}
The correspondence above is a formal analogy, not a mathematical reduction. The two columns share the same abstract structure (the
observational order) but the physical column does not follow
mathematically from the computer-science column without the additional
ingredients noted in the preceding remark.
\end{remark}

\section{Observer-parametrised complexity classes}
\label{sec:pvsnp}

Classical complexity theory implicitly assumes the observer is always $O_\top$. The observational hierarchy makes this assumption explicit and
allows it to be relaxed.

\begin{remark}[Computational status of observers]
\label{rem:comp-status}
An observer $O : \Sigma^* \to S$ is a fixed total preprocessing map applied once to the input before any machine computation begins. It does
not answer queries, does not interact with the machine, and does not extend the computational model. All machines in this paper are standard
Turing machines over $\Sigma^*$; receiving $O(x)$ means receiving
$\mathrm{enc}(O(x))$ on the input tape (Convention~\ref{conv:encoding}).
The observer only reduces available information and never increases
computational power.

A language $L$ that is not $O$-saturated (Definition~\ref{def:saturation})
is not merely hard under observer $O$: it is not well-posed in the
observation space of $O$. If $x \sim_O y$ but $x \in L$ and $y \notin L$,
then any machine receiving only $O(x)$ cannot decide correctly on both inputs, regardless of computational resources. This is \emph{structural
blindness}, not computational hardness.
\end{remark}

\begin{lemma}[Simulation for computable observers]
\label{lem:simulation}
Let $O : \Sigma^* \to S$ be computable in time $s(|x|)$, where the time bound includes the steps required to write $\mathrm{enc}(O(x))$ on the
output tape, so that $|\mathrm{enc}(O(x))| = O(s(|x|))$. If a Turing
machine decides a language $L$ on input $\mathrm{enc}(O(x))$ in time $p(|\mathrm{enc}(O(x))|)$ for some polynomial $p$, then there exists a
Turing machine deciding $L$ on input $x$ in time
\[
  q(|x|) = s(|x|) + p(O(s(|x|))) + O(1),
\]
which is polynomial in $|x|$ whenever $s$ is polynomial. In particular, $\mathbf{P}_O \subseteq \mathbf{P}$ for every observer computable in
polynomial time.
\end{lemma}

\begin{proof}[Proof sketch]
On input $x$, the simulating machine first computes $\mathrm{enc}(O(x))$ in time $s(|x|)$; the output length satisfies
$|\mathrm{enc}(O(x))| = O(s(|x|))$ by the time bound on the
computation. It then runs the given machine on $\mathrm{enc}(O(x))$,
which takes time $p(|\mathrm{enc}(O(x))|) = p(O(s(|x|)))$. The total time $q(|x|) = s(|x|) + p(O(s(|x|))) + O(1)$ is polynomial in $|x|$
whenever $s$ is polynomial. Full proof in Appendix~\ref{app:simulation}.
\end{proof}

\begin{remark}[On the boldness of this formalisation]
We are aware that Definition~\ref{def:PO_NPO} below is a non-standard step: defining complexity classes parametrised by an observer opens a
two-dimensional landscape in which $\mathbf{P} = \mathbf{NP}$ is one
point. We offer it because the calculations compel us to. The
incomparability results of Section~\ref{sec:gerarchia}, the orthogonality
with the Chomsky hierarchy in Section~\ref{sec:ortogonalita}, and the
collapse of $\mathbf{P}_{O_{\mathrm{prof}}} =
\mathbf{NP}_{O_{\mathrm{prof}}}$ established below are all concrete, verifiable facts. They do not prove anything about the classical
$\mathbf{P}$ vs $\mathbf{NP}$ question, and we do not claim they do. What they show is that the distinction between hardness and blindness is
real and measurable.
\end{remark}

\begin{definition}[$\mathbf{P}_O$ and $\mathbf{NP}_O$]
\label{def:PO_NPO}
Fix an observer $O : \Sigma^* \to S$.
\begin{itemize}
  \item $\mathbf{P}_O$ is the set of languages decidable in time 
        polynomial in $|\mathrm{enc}(O(x))|$ by a deterministic machine
        receiving $\mathrm{enc}(O(x))$ as its entire input.
  \item $\mathbf{NP}_O$ is the set of languages for which there exists
        a nondeterministic machine that, receiving $\mathrm{enc}(O(x))$
        and a certificate of length polynomial in $|\mathrm{enc}(O(x))|$,
        verifies membership in time polynomial in $|\mathrm{enc}(O(x))|$.
\end{itemize}
\end{definition}

\begin{remark}
The complexity is measured with respect to $|\mathrm{enc}(O(x))|$, the
size of the input actually received by the machine, not $|x|$. For $O = O_\top$ these coincide. For $O = O_{\mathrm{prof}}$ on a binary
alphabet, $\mathrm{enc}(O_{\mathrm{prof}}(x))$ encodes the pair $(|x|_0, |x|_1)$ in binary, giving
$|\mathrm{enc}(O_{\mathrm{prof}}(x))| = O(\log |x|)$, strictly smaller than $|x|$.

The machine receives $\mathrm{enc}(O(x))$ as given and is not required
to compute $O(x)$ from $x$. When reducing $\mathbf{P}_O$ to $\mathbf{P}$,
however, one must construct a standard TM that first computes $O(x)$ from
$x$ and then runs the given machine: this requires $O$ to be computable,
as formalised in Lemma~\ref{lem:simulation}.
All canonical observers of Definition~\ref{def:osservatori_canonici} are
computable; in particular, $O_{\mathrm{prof}}(x) = (|x|_0, |x|_1)$ is
computed by a single scan of $x$.
\end{remark}

\begin{proposition}[Boundary cases]
\label{prop:boundary}
\begin{enumerate}[label=(\roman*)]
  \item $\mathbf{P}_{O_\top} = \mathbf{P}$ and
        $\mathbf{NP}_{O_\top} = \mathbf{NP}$.
  \item $\mathbf{P}_{O_\bot} = \mathbf{NP}_{O_\bot} =
        \{\emptyset, \Sigma^*\}$.
  \item $\mathbf{P}_{O_{\mathrm{prof}}} = \mathbf{NP}_{O_{\mathrm{prof}}}$,
        consisting exactly of the $O_{\mathrm{prof}}$-saturated languages
        $L = \{x \mid f(\prof(x)) = 1\}$ for which $f : \NN^k \to \{0,1\}$
        is computable in time polynomial in
        $|\mathrm{enc}(O_{\mathrm{prof}}(x))|$.
\end{enumerate}
\end{proposition}

\begin{proof}
(i) $O_\top(x) = x$, so a machine in $\mathbf{P}_{O_\top}$ receives $x$
directly and runs in time polynomial in $|x|$: this is $\mathbf{P}$.
The same for $\mathbf{NP}$.

(ii) The observation space of $O_\bot$ is the singleton $\{\star\}$:
every input maps to the same value. A deterministic or nondeterministic
machine receiving always $\star$ either accepts all strings or none.

(iii) A machine receiving $\mathrm{enc}(O_{\mathrm{prof}}(x)) \in
\{0,1\}^*$ has a complete description of the symbol counts. Membership in $L$ depends (by $O_{\mathrm{prof}}$-saturation) only on this vector.
The machine reads its entire input in time linear in $|\mathrm{enc}(O_{\mathrm{prof}}(x))|$ and evaluates the membership
predicate. There is no combinatorial structure remaining to search over, the observer has performed the only aggregation possible, so
nondeterminism adds nothing. Hence $\mathbf{P}_{O_{\mathrm{prof}}} =
\mathbf{NP}_{O_{\mathrm{prof}}}$.
\end{proof}

\begin{proposition}[Position of $\mathbf{P}_{O_{\mathrm{prof}}}$]
\label{prop:Oprof-vs-P}
\begin{enumerate}[label=(\roman*)]
  \item $\mathbf{P}_{O_{\mathrm{prof}}} \subseteq \mathbf{P}$.
  \item $\mathbf{P}_{O_{\mathrm{prof}}} \subsetneq \mathbf{P}$.
  \item $\mathbf{P}_{O_{\mathrm{prof}}} = \mathbf{NP}_{O_{\mathrm{prof}}}
        \subsetneq \mathbf{P}$.
\end{enumerate}
\end{proposition}

\begin{proof}
\emph{(i)} By Lemma~\ref{lem:simulation}: $O_{\mathrm{prof}}$ is computable in $O(|x|)$ time (one left-to-right scan of $x$), so
Lemma~\ref{lem:simulation} applies and gives $\mathbf{P}_{O_{\mathrm{prof}}} \subseteq \mathbf{P}$.

Concretely: given $L \in \mathbf{P}_{O_{\mathrm{prof}}}$, let $M$ be a TM deciding $L$ on input $\mathrm{enc}(O_{\mathrm{prof}}(x))$ in time
$T(|\mathrm{enc}(O_{\mathrm{prof}}(x))|)$ for some polynomial $T$. Construct $M'$ that on input $x$: (1) computes $\mathrm{enc}
(O_{\mathrm{prof}}(x)) = \mathrm{enc}(|x|_0, |x|_1)$ in $O(|x|)$
steps; (2) runs $M$ on this encoding. Since $|\mathrm{enc}(O_{\mathrm{prof}}(x))| = O(\log|x|)$, step~(2) takes $T(O(\log|x|)) = o(|x|)$ steps. Total: $O(|x|) \in \mathbf{P}$.

\emph{(ii)} The language $0^*1^*$ is in $\mathbf{P}$ (one linear scan)
but not $O_{\mathrm{prof}}$-saturated (Theorem~\ref{thm:ortogonalita}),
so $0^*1^* \notin \mathbf{P}_{O_{\mathrm{prof}}}$.

\emph{(iii)} Combines Proposition~\ref{prop:boundary}(iii) with (ii).
\end{proof}

\begin{remark}[Implications of the structural collapse]
\label{rem:Oprof-vs-P}
Proposition~\ref{prop:Oprof-vs-P}(iii) has four distinct implications.

\medskip\noindent
\textbf{(a) Non-triviality.}
$\mathbf{P}_{O_{\mathrm{prof}}}$ is a genuine intermediate class:
nonempty and strictly below $\mathbf{P}$.

\medskip\noindent
\textbf{(b) Separation of hardness from blindness.}
The collapse $\mathbf{P}_{O_{\mathrm{prof}}} =
\mathbf{NP}_{O_{\mathrm{prof}}}$ is not equivalent to
$\mathbf{P} = \mathbf{NP}$. The proof makes no assumption about
$\mathbf{P}$ vs $\mathbf{NP}$ and implies nothing about it. The collapse
is due entirely to structural blindness: $O_{\mathrm{prof}}$ discards
so much information that no combinatorial structure remains for
nondeterminism to exploit.

\medskip\noindent
\textbf{(c) Sublinear containment.}
Since $|\mathrm{enc}(O_{\mathrm{prof}}(x))| = O(\log|x|)$, every machine operating under $O_{\mathrm{prof}}$ receives a logarithmically
compressed summary. The total runtime of the simulation in Proposition~\ref{prop:Oprof-vs-P}(i) is $O(|x|)$, dominated by the
cost of computing $O_{\mathrm{prof}}(x)$, while the decision itself takes $o(|x|)$. The collapse occurs strictly below $\mathbf{P}$ in
terms of the original input size.

\medskip\noindent
\textbf{(d) Conditional position in the $\mathbf{P}$ vs $\mathbf{NP}$
landscape.}
Proposition~\ref{prop:Oprof-vs-P}(iii) holds unconditionally. If one additionally assumes $\mathbf{P} \neq \mathbf{NP}$ (a hypothesis
external to this paper, neither assumed nor implied by any result here), then with Proposition~\ref{prop:boundary}(i):
\[
  \mathbf{NP}_{O_{\mathrm{prof}}}
  = \mathbf{P}_{O_{\mathrm{prof}}}
  \subsetneq \mathbf{P}
  = \mathbf{P}_{O_\top}
  \subsetneq \mathbf{NP}
  = \mathbf{NP}_{O_\top}.
\]
Under this conditional hypothesis, every language a nondeterministic
verifier with observer $O_{\mathrm{prof}}$ can certify is already
decidable in polynomial time by a deterministic machine with full input
access. The reason is not that nondeterminism loses its power, but that
$O_{\mathrm{prof}}$ cannot see the combinatorial structure that makes
problems hard in the first place.
\end{remark}

\begin{theorem}[Two independent axes]
\label{thm:due_assi}
Computational hardness and observational decidability are independent:
\begin{enumerate}[label=(\roman*)]
  \item There exist languages in $\mathbf{P}$ that are not
        $O_{\mathrm{prof}}$-saturated.
  \item There exist languages outside $\mathbf{RE}$ that are
        $O_{\mathrm{prof}}$-saturated.
\end{enumerate}
\end{theorem}

\begin{proof}
(i) $0^*1^*$ is in $\mathbf{P}$ but not $O_{\mathrm{prof}}$-saturated
(Theorem~\ref{thm:ortogonalita}).

(ii) Let $g : \NN^2 \to \{0,1\}$ be defined by $g(m,n) = 1$ if and only
if the $m$-th Turing machine does not halt on input $n$. The language
$L = \{x \in \{0,1\}^* \mid g(|x|_0, |x|_1) = 1\}$ is
$O_{\mathrm{prof}}$-saturated by construction. It is outside
$\mathbf{RE}$: the map $(m,n) \mapsto 0^m 1^n$ would witness
$K^c \leq_m L$ if $L$ were r.e.\ (since $0^m 1^n \in L \iff g(m,n) = 1
\iff (m,n) \in K^c$), making $K^c$ r.e., a contradiction since $K^c$
is the complement of the halting set.
\end{proof}

\begin{corollary}[Four-quadrant independence]
\label{cor:quattro_quadranti}
All four combinations of membership in $\mathbf{P}$ and
$O_{\mathrm{prof}}$-saturation are non-empty:
\begin{center}
\begin{tabular}{lll}
  \toprule
  & $O_{\mathrm{prof}}$-saturated & not $O_{\mathrm{prof}}$-saturated \\
  \midrule
  in $\mathbf{P}$ &
    $\{x \mid |x|_1 \text{ even}\}$ &
    $0^*1^*$ \\
  outside $\mathbf{RE}$ &
    as in Theorem~\ref{thm:due_assi}(ii) &
    $L_{\mathrm{br}}$ (see below) \\
  \bottomrule
\end{tabular}
\end{center}
\end{corollary}

\begin{proof}
The top row and bottom-left cell follow from
Theorem~\ref{thm:ortogonalita} and Theorem~\ref{thm:due_assi}.

For the bottom-right cell, let $K^c = \{n \in \NN \mid \text{the }
n\text{-th TM does not halt on input } n\}$.
Since $K^c$ is nonempty, fix $m \in K^c$. Define
\[
  L_{\mathrm{br}} = \{x \in \{0,1\}^* \mid
    x \text{ starts with } 0 \text{ and } |x|_0 \in K^c\}.
\]
$L_{\mathrm{br}}$ is not $O_{\mathrm{prof}}$-saturated: the strings
$0^m 1$ and $10^m$ both have profile $(m,1)$, but $0^m 1$ starts with
$0$ and $|0^m 1|_0 = m \in K^c$ so $0^m 1 \in L_{\mathrm{br}}$, while
$10^m$ starts with $1$ so $10^m \notin L_{\mathrm{br}}$.

$L_{\mathrm{br}}$ is outside $\mathbf{RE}$: if $L_{\mathrm{br}}$ were
r.e., the map $n \mapsto 0^n$ would witness $K^c \leq_m L_{\mathrm{br}}$
(since $0^n$ starts with $0$ and $|0^n|_0 = n \in K^c \iff 0^n \in
L_{\mathrm{br}}$), making $K^c$ r.e., a contradiction.
\end{proof}

\section{Conclusion}

We have introduced the observational hierarchy as a new axis of
classification for formal languages, independent of the Chomsky
hierarchy. The hierarchy has the structure of a partial order with a
diamond-shaped profile sub-lattice comprising the length branch
$O_\bot \prec O_{\mathrm{len}} \prec O_{\mathrm{prof}} \prec O_\top$
and the parity branch
$O_\bot \prec O_{\mathrm{par}} \prec O_{\mathrm{prof}} \prec O_\top$,
with $O_{\mathrm{len}}$ and $O_{\mathrm{par}}$ incomparable at the
intermediate level, and an infinite subsequence branch
$O_\bot \prec O_1 \prec O_2 \prec \cdots \prec O_\top$.

We have shown that computational limitations depend not only on the
power of the machine, but on the structure of the observer. There exist
languages that no Turing machine can recognise when equipped with a
sufficiently weak observer, but that a DFA recognises with a more
powerful observer. The power of the machine cannot compensate for the
weakness of the observer.

The framework extends to complexity theory through
$\mathbf{P}_O$ and $\mathbf{NP}_O$. The central distinction is between
computational hardness (which depends on the machine) and structural
blindness (which depends on what the observer has discarded). These are
two independent phenomena. In particular,
$\mathbf{P}_{O_{\mathrm{prof}}} = \mathbf{NP}_{O_{\mathrm{prof}}}$
holds as a structural collapse strictly inside $\mathbf{P}$
(Proposition~\ref{prop:Oprof-vs-P}): it is not a resolution of the
classical $\mathbf{P}$ vs $\mathbf{NP}$ question, but evidence that
blindness and hardness are genuinely different phenomena.

\paragraph{Open directions.}
\begin{enumerate}[label=(\roman*)]
  \item \emph{Probabilistic observers}: stochastic channels
        $O : \Sigma^* \to \Delta(S)$, where $\preceq$ becomes Blackwell
        degradation (Section~\ref{sec:teoria} sketches the connection).
  \item \emph{Composed observers}: categorical structure of observer
        composition, its relation to forgetful functors, and the
        connection between the product of two observers and their join in
        the observational lattice.
  \item \emph{Observational complexity and information}: formalising the
        connection between the observational complexity of a language and
        the amount of information the observer must preserve, via the
        data processing inequality.
  \item \emph{Intermediate observers}: algebraic characterisation of
        classes induced by observers that combine subsequence and count
        information, lying strictly between the profile sub-lattice and
        the subsequence branch.
  \item \emph{Profile sub-lattice completion}: characterisation of all
        observers strictly between $O_\bot$ and $O_{\mathrm{prof}}$,
        incomparable with both $O_{\mathrm{len}}$ and $O_{\mathrm{par}}$.
  \item \emph{Sublinear complexity}: whether $\mathbf{P}_{O_{\mathrm{prof}}}$
        coincides with a known sublinear complexity class, and how the
        observational axis interacts with the fine structure of complexity
        below $\mathbf{P}$ (Remark~\ref{rem:Oprof-vs-P}(c)).
\end{enumerate}

\begin{remark}
Directions (ii) and (iii) are closely related: a categorical
formalisation of (ii) would provide the infrastructure for the
quantitative theory in (iii). Direction (vi) is directly connected to
Remark~\ref{rem:Oprof-vs-P}(c): any progress on the position of
$\mathbf{P}_{O_{\mathrm{prof}}}$ within sublinear complexity would
sharpen our understanding of the structural collapse in
Proposition~\ref{prop:boundary}(iii).
\end{remark}

\paragraph{Personal note.}
%
%
Several sources shaped the conceptual framework developed here.
Jacques Monod's \textit{Le Hasard et la N\'ecessit\'e} (1970) showed
how teleonomic properties are invisible to the molecular chemistry that
implements them. Jorge Luis Borges's \textit{La Biblioteca de Babel},
Douglas Hofstadter's \textit{G\"odel, Escher, Bach} (1979), and
John D.\ Barrow's \textit{Impossibility: The Limits of Science and the
Science of Limits} (1998) contributed to the imaginative framework that
crystallised into the ideas developed here.\footnote{These sources
shaped the conceptual framework. All formal results are independent of
them and stand on their own proofs.}

\section*{Acknowledgments}
The author used an artificial intelligence based language assistant to
support text revision, translation, and bibliography formatting. All
scientific ideas and conclusions are the author's own.

\nocite{*}
\bibliographystyle{plain}
\bibliography{references}

\appendix

\section{Proof of Lemma~\ref{lem:simulation}}
\label{app:simulation}

This appendix gives the full formal proof of the simulation lemma stated
and sketched in Section~\ref{sec:pvsnp}.

\begin{proof}[Full proof of Lemma~\ref{lem:simulation}]
Let $O : \Sigma^* \to S$ be computable in time $s(|x|)$ where $s$ is a
function such that the computation halts after at most $s(|x|)$ steps
and writes $\mathrm{enc}(O(x))$ on a dedicated output tape.
Since a Turing machine can write at most one symbol per step, the output
length satisfies $|\mathrm{enc}(O(x))| \leq s(|x|)$, hence
$|\mathrm{enc}(O(x))| = O(s(|x|))$.

Let $M$ be a TM that decides $L$ on input $\mathrm{enc}(O(x))$ in time
$p(|\mathrm{enc}(O(x))|)$ for some polynomial $p$. Construct the
simulating machine $M'$ as follows. On input $x \in \Sigma^*$:
\begin{enumerate}
  \item Run the computation of $O$ on $x$ to produce
        $y = \mathrm{enc}(O(x))$. This takes at most $s(|x|)$ steps.
  \item Run $M$ on $y$. This takes at most $p(|y|)$ steps.
  \item Output whatever $M$ outputs.
\end{enumerate}
The total time for $M'$ on $x$ is at most $s(|x|) + p(|y|) + O(1)$.
Since $|y| = |\mathrm{enc}(O(x))| \leq s(|x|)$, we have
$p(|y|) \leq p(s(|x|))$, giving total time
$q(|x|) = s(|x|) + p(s(|x|)) + O(1)$.
When $s$ is a polynomial, $q$ is also a polynomial in $|x|$, so
$M' \in \mathbf{P}$ witnesses $L \in \mathbf{P}$.

The conclusion $\mathbf{P}_O \subseteq \mathbf{P}$ follows: every
$L \in \mathbf{P}_O$ has a TM $M$ deciding it under $O$ in polynomial
time in $|\mathrm{enc}(O(x))|$. The construction above produces $M'$
deciding $L$ under $O_\top$ in polynomial time in $|x|$.
\end{proof}

\end{document}